\documentclass{ws-fnl}
\usepackage{cite}
\begin{document}

\markboth{Authors' Names}{ (Paper's Title)}

\catchline{}{}{}{}{}

\title{Multifractal of mass function}

\author{Chenhui Qiang}

\address{Institute of Fundamental and Frontier Science, University of Electronic Science and Technology of China, 610054 Chengdu, China \\ Yingcai Honors College, University of Electronic Science and Technology of China, Chengdu, 610054, China\\
2019270101007@std.uestc.edu.cn}

\author{Yong Deng\footnote{Corresponding author}}

\address{Institute of Fundamental and Frontier Science, University of Electronic Science and Technology of China, 610054 Chengdu, China \\
School of Education, Shaanxi Normal University, Xi'an, 710062, China \\
School of Knowledge Science, Japan Advanced Institute of Science and Technology, Nomi, Ishikawa 923-1211, Japan \\
Department of Management, Technology, and Economics, ETH Zürich, Zurich, Switzerland \\
dengentropy@uestc.edu.cn}

\maketitle

\begin{history}
\received{(received date)}
\revised{(revised date)}
\end{history}
\begin{abstract}
Multifractal plays an important role in many fields. However, there is few attentions about mass function, which can better deal with uncertain information than probability. In this paper, we proposed multifractal of mass function. Firstly, the definition of multifractal spectrum of mass function is given. Secondly, the multifractal dimension of mass function is defined as $D_{\alpha}$. When mass function degenerates to probability distribution,  $D_{\alpha}$ degenerates to $d_{\alpha}$, which is information dimension proposed by Renyi. One interesting property is that the multifractal dimension of mass function with maximum Deng entropy is 1.585 no matter the order. Other interesting properties and numerical examples are shown to illustrate proposed model.

\keywords{Multifractal; Mass function; Renyi information dimension; Deng entropy}
\end{abstract}

\section{Introduction}
In recently years, much research has focus on fractal theory\ \cite{Mandelbrot1998The,kiew2020analysis} since many natural phenomenon\ \cite{yu2019landscape,liu2019integrated} and system\ \cite{harabagiu2009particle,namazi2018fractal} can be characterized by the fractal properties. Thus, fractal plays a vital role in many fields such as mechanical engineering\ \cite{zhao2016application,fang2019simulation}, geotechnical engineering\ \cite{li2021permeability,li2019characterization,guo2019anti}, oscillator model\ \cite{wang2021new,elias2021equivalent}, molecular structure\ \cite{siddiqui2019molecular,ding2017question} and so on. Lots of models about fractal dimension are proposed\ \cite{wen2021fractal,gires2017fractal}, which is the main parameter of measuring irregular objects. To better analyse the property of fractal, many models about multifractal\ \cite{lopes2009fractal,landais2019multifractal,sanyal2021tagore} were proposed, which as a generalization of fractal can better describe the variation of local features. For example, the characteristics of multifractal spectrum was analysed to study the stability of the China's stock market\ \cite{li2020institutional}. A generalization of classical multifractal detrended fluctuation analysis was proposed by Wang\ \cite{wang2021multifractal}. A method named multifractal cross wavelet analysis was developed to characterizes the properties of complex system\ \cite{jiang2017multifractal}.

How to deal with uncertain information has attracted a lot of attention\ \cite{wang2017rumor,Tsallis2021,cheong2019paradoxical}. Many models like probability theory\ \cite{liu2021probability,jiang2019new}, fuzzy sets\ \cite{senapati2020fermatean,kutlu2019spherical,xiao2021caftr}, evidence theory\ \cite{Xiao2021CEQD,Deng2020InformationVolume,zhou2021counter}, rough sets\ \cite{yan2018short,alcantud2019n} are developed. Mass function is a significant component of evidence theory\ \cite{dempster1967upper,shafer1976mathematical}  and has an advantage over probability distribution in dealing with uncertainty problem\ \cite{chen2021fire}. Lots of parameters of mass function were studied like divergence\ \cite{wang2021new1}, correlation coefficient\ \cite{jiang2018correlation,Qiang2021A}, distance\ \cite{xiao2020ced,han2016belief}. In addition, the uncertainty of mass function has been studied extensively\ \cite{Xue2021Interval,khalaj2020new}. The negation of mass function was developed by Gao\ \cite{gao2021generating} and Mao\ \cite{Mao2021IntervalNegation}. Some effective combination rule of conflict mass function were proposed\ \cite{Xiong2021InformationSciences,huang2020evidential,song2020self}. Transform mass function into probability is still a hot topic\ \cite{CHEN2021104438,huang2021basic}. However, there is few attentions paid to mass function from the perspective of multifractal. In this paper, the multifractal analysis of mass function are given including multifractal spectrum and multifractal dimension.

This article is organized as follows. Section 2 is a brief introduction about preliminaries. Multifractal spectrum of mass function is showed in Section 3. In Section 4, the definition of multifractal dimension is given. Detailed calculation and numerical examples are following. Finally is a simple conclusion for whole work. 

\section{Preliminaries}

\subsection{Multifractal spectrum}

Consider a measure space $(\chi,B(\chi),\mu )$. There is a lattice covering of $\chi$ by $d-$dimensional boxes of width $\delta _{n}$, where $B_{\delta _{n}}(x)$ is the box that contains the point $x$\ \cite{harte2001multifractals}. $U_{n}(x)$ is function with the condition $\delta _{n}\rightarrow 0$ as $n\rightarrow \infty $\ \cite{harte2001multifractals}.
\begin{equation}
    U_{n}(x)=-log\mu\left [B_{\delta _{n}}(x) \right ],\ if\  \mu\left [B_{\delta _{n}}(x) \right ] > 0\,. \label{1}
\end{equation}

Let $Y_{n}(x)$ be a rescaled version of $U_{n}(x)$ with the condition $\lim_{n\rightarrow \infty }Pr\left \{Y_{n}=y\right \}=0$ if $y\neq  y_{0}$\ \cite{harte2001multifractals}. 

\begin{equation}
     Y_{n}(x)=\frac{U_{n}(x)}{-log\delta_{n}}\,. \label{2}
\end{equation}

The multifractal spectrum is defined as follows\ \cite{harte2001multifractals}.
\begin{equation}
        \widetilde{f}(y)=\lim_{n\rightarrow \infty }\frac{logN_{n}}{-log\delta _{n}}+\lim_{\epsilon \rightarrow 0}\lim_{n \rightarrow 0}\frac{Pr\left \{|Y_{n}-y|<\epsilon \right \} }{-log\delta _{n}}\,. \label{3}
\end{equation}

Where $N_{n}$ is the number of boxes at the $n$th stage with positive $\mu$ measure\ \cite{harte2001multifractals}. The first term is the box counting dimension of the support of $\mu$ and the second term is the rate of the probability function $Y_{n}$ approaches zeros\ \cite{harte2001multifractals}.

\subsection{Renyi entropy and Renyi information dimension}
\subsubsection{Renyi entropy}

Entropy is an important measure in many fields\ \cite{babajanyan2020energy,Song2021powerset}. A classic entropy named Renyi entropy is defined as follows\ \cite{van2014renyi}.

Consider a probability distribution $P_{X}$ of a discrete random variable $X$: $P_{X}=\left \{p_{i}|i=1,2,...,n\right \}$, its Renyi entropy is,
\begin{equation}
    H_{\alpha}(P)=\frac{1}{1-\alpha}log\sum_{i}^{n}p_{i}^{\alpha}\,. \label{4}
\end{equation}
Where $\alpha \geq 0$. When $\alpha \rightarrow 1$, Renyi entropy degenerates to Shannon entropy\ \cite{lassance2021minimum}.

\begin{equation}
    \lim_{\alpha\rightarrow 1}H_{\alpha}(P)=\sum _{i}^{n}p_{i}logp_{i}=H_{S}\,. \label{5}
\end{equation}

\subsubsection{Renyi information dimension}
The dimension of the probability distribution of $X$ defined by Renyi is as follows\ \cite{duan2019new}.
\begin{equation}
    d_{\alpha}=\lim_{n\rightarrow \infty }\frac{H_{\alpha}(P)}{log_{2}n}\,. \label{6}
\end{equation}
When $\alpha \rightarrow 1$, $d_{\alpha}$ represents the rate of Shannon entropy grows with scale.
\subsection{Mass function}
Mass function assigns mass on power set and its definition is as follows.

Let $\Theta$ is the framework of discernment in evidence theory\ \cite{dempster1967upper,shafer1976mathematical},
\begin{equation}
\Theta=\{\theta _1,\cdots,\theta _n\}, n\in N^{+} \,. \label{7}
\end{equation}

Its power set is,
\begin{equation}
\begin{split}
2^{\Theta}=&\{A_{1},\ A_{2},...,\ A_{2^{n}}\} \\
=&\{\{\theta_1\},\cdots,\ \{\theta_1,\theta_{2}\},\cdots,\ \{\Theta\}, \ \emptyset\},\,. \label{8}
\end{split}
\end{equation}

Mass function is a mapping $ m:2^{\Theta}\rightarrow [0,1]$ satisfing\ \cite{dempster1967upper,shafer1976mathematical},

\begin{equation}
    m(\emptyset)=0,\, \label{9}
\end{equation}
\begin{equation}
   \quad \sum_{A_{i}\in 2^{\Theta}}m(A_{i})=1\,. \label{10}
\end{equation}

where $A_{i}$ is called focal element when $m(A_{i})>0$.

\subsection{Deng entropy}
Given a framework of discernment is $\Theta$ and a mass function is $m(A_{i})$. Deng entropy is defined as follows\ \cite{deng2020uncertainty}.
\begin{equation}
     H_{D}=-\sum_{A_{i} \in 2^{\Theta}}m(A_{i})log(\frac{m(A_{i})}{2^{|A_{i}|}-1})\,. \label{11}
\end{equation}

Where $|A_{i}|$ is the cardinal of focal element. When mass function satisfies the condition,
\begin{equation}
    m(A_{i})=\frac{2^{|A_{i}|}-1}{\sum _{i}2^{|A_{i}|}-1},\, \label{12}
\end{equation} 

Deng entropy reaches the maximum\ \cite{deng2020uncertainty}.
\begin{equation}
    H_{maxD}=-\sum_{i}m(A_{i})log(\frac{m(A_{i})}{2^{|A_{i}|}-1})=log\sum_{i}({2^{|A_{i}|}-1})\,. \label{13}
\end{equation}

Deng entropy considers the non-specificity of mass function and is studied widely\ \cite{Liao2020,DengeXtropy,kazemi2021fractional}.

 

\section{Multifractal spectrum of mass function}
We first generalizes the concept of multifractal spectrum of mass function, then some examples are shown to illustrate the proposed model.
\begin{definition}
A measure space in evidence theory is $(\Theta, 2^{\Theta}, m)$. Where $m$ is a mass function of $2^{\Theta}$. The multifractal spectrum is defined as follows.
\begin{equation}
    \widetilde{f}(y)=\frac{log_{2}N_{n}}{-log_{2}(2^{|\Theta|}-1)^{-1}}\,. \label{14}
\end{equation}

Where $N_{n}$ is the number of focal elements that have the same mass and $y$ is calculated as follows. 

\begin{equation}
    y(A_{i})=\frac{-log_{2}m(A_{i})}{-log_{2}(2^{|\Theta|}-1)^{-1}}\,. \label{15}
\end{equation}

$y$ can be explained as the rate at which mass function changes with respect to the size of $\Theta$. $\widetilde{f(y)}$ is a function of $y$ and following is a detailed process.
\end{definition}

\noindent \textbf {Example 1}: Consider a framework of discernment is $\Theta=\left \{\theta_{1},\theta_{2},\theta_{3} \right \}$. A mass function assignment with maximum Deng entropy is $m(\theta_{1})=m(\theta_{2})=m(\theta_{3})=1/19$, $m(\theta_{1},\theta_{2})=m(\theta_{1},\theta_{3})=m(\theta_{2},\theta_{3})=3/19$, $m(\theta_{1},\theta_{2},\theta_{3})=7/19$. According to Eq.~(\ref{14}) and Eq.~(\ref{15})

\begin{equation}
    y(\theta_{1})= y(\theta_{2})= y(\theta_{3})=\frac{-log_{2}(1/19)}{-log_{2}(2^{3}-1)^{-1}}=1.5131\,. \label{16}
\end{equation}
\begin{equation}
   y(\theta_{1},\theta_{2})=y(\theta_{1},\theta_{3})=y(\theta_{2},\theta_{3})=\frac{-log_{2}(3/19)}{-log_{2}(2^{3}-1)^{-1}}=0.9486\,. \label{17}
\end{equation}
\begin{equation}
    y(\theta_{1},\theta_{2},\theta_{3})=\frac{-log_{2}(1/19)}{-log_{2}(2^{3}-1)^{-1}}=0.5131\,. \label{18}
\end{equation}
\begin{equation}
    \widetilde{f}(y(\theta_{1}))= \widetilde{f}(y(\theta_{2}))=\widetilde{f}( y(\theta_{3}))=\frac{log_{2}(3)}{-log_{2}(2^{3}-1)^{-1}}=0.5646\,. \label{19}
\end{equation}
\begin{equation}
   \widetilde{f}(y(\theta_{1},\theta_{2}))=\widetilde{f}(y(\theta_{1},\theta_{3}))=\widetilde{f}(y(\theta_{2},\theta_{3}))=\frac{log_{2}(3)}{-log_{2}(2^{3}-1)^{-1}}=0.5646\,. \label{20}
\end{equation}
\begin{equation}
    \widetilde{f}(y(\theta_{1},\theta_{2},\theta_{3}))=\frac{log_{2}(1)}{-log_{2}(2^{3}-1)^{-1}}=0\,. \label{21}
\end{equation}

From above calculations, if the cardinal of the focal elements is equal, their mass functions are the same and have the same $y$. These pairs of points are drawn in Fig. 1.
\begin{figure}[th]
\centerline{\includegraphics[width=10cm]{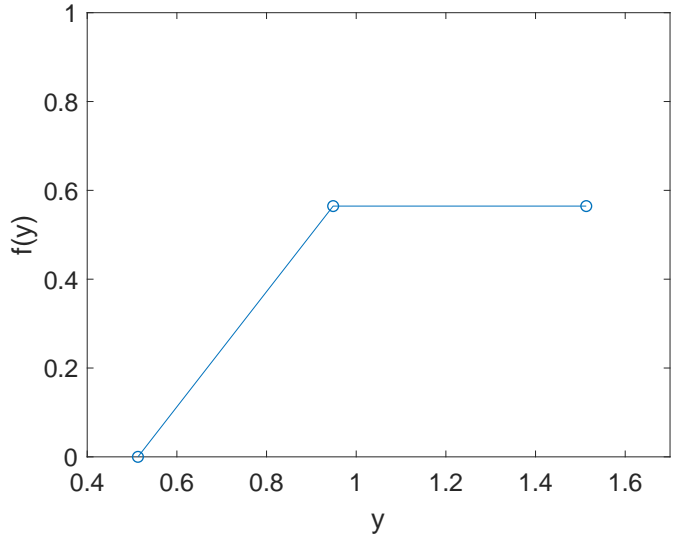}}
\vspace*{8pt}
\caption{The multifractal spectrum of Example 1}
\end{figure}

\noindent \textbf {Example 2}: Consider a framework of discernment is $\Theta$, $|\Theta|=2,...,n$. A series of mass functions assigned with maximum Deng entropy are $m_{i}$, $i=2,...,n$. The multifractal spectrums are shown in Fig. 2. and Fig. 3. Part of the values of $y$ and $\widetilde{f}(y)$ are shown in Table 1 and Table 2.
\begin{figure}[th]
\centerline{\includegraphics[width=10cm]{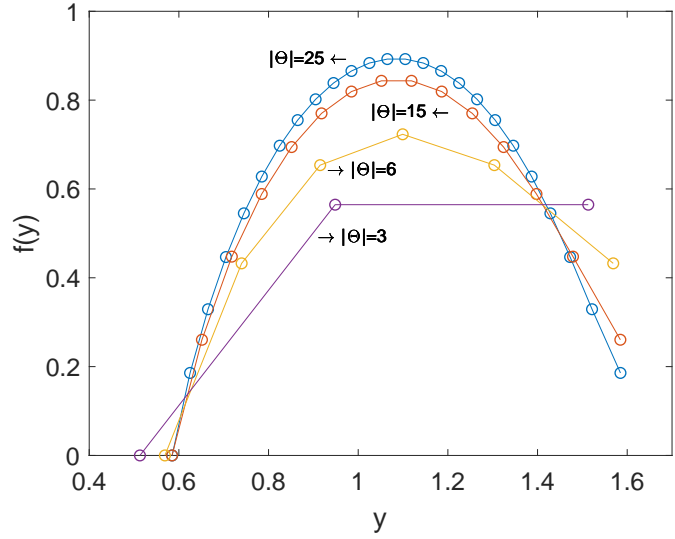}}
\vspace*{8pt}
\caption{The multifractal spectrum of Example 2 when $|\Theta|=3,6,15,25$}
\end{figure}

\begin{figure}[th]
\centerline{\includegraphics[width=10cm]{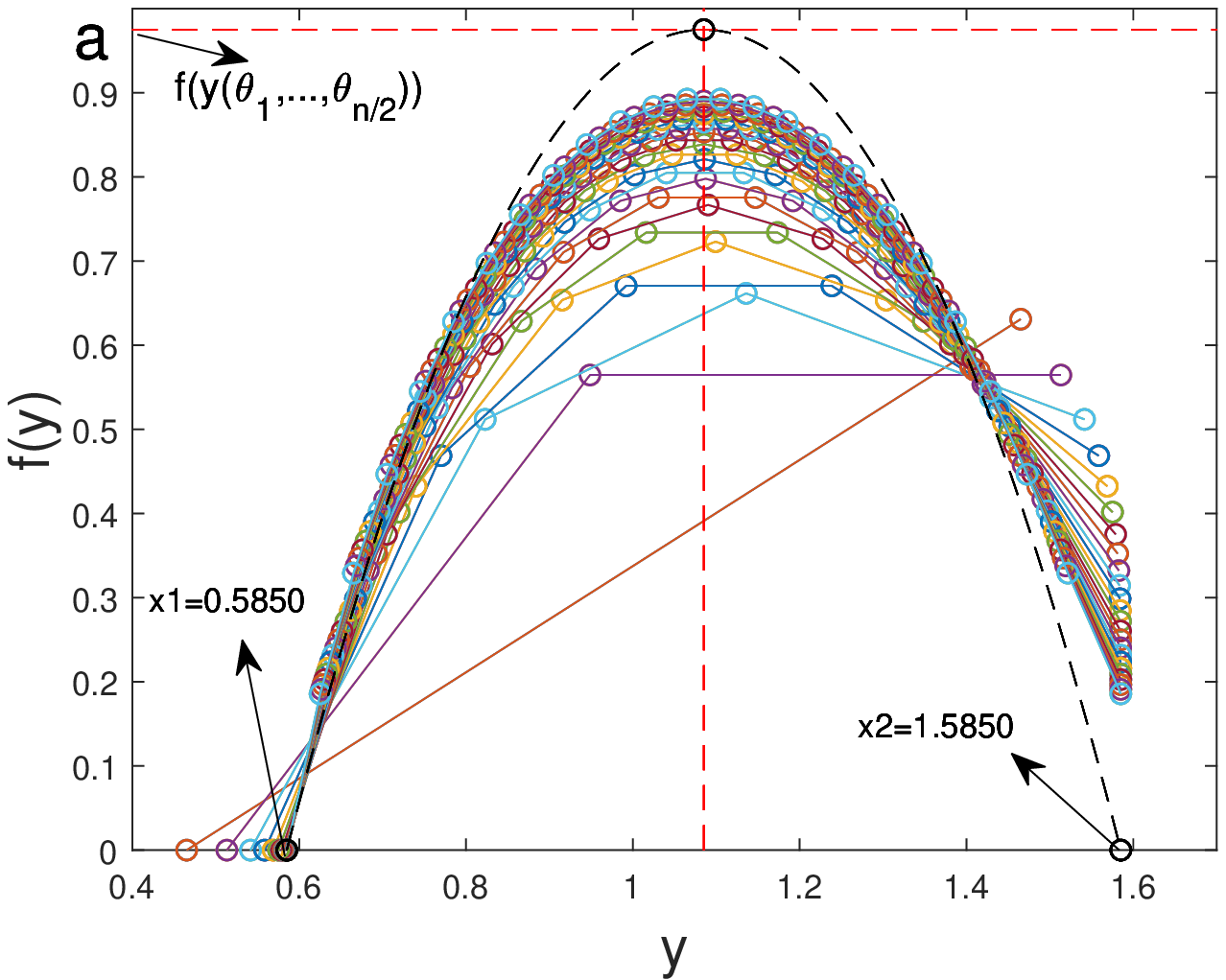}}
\vspace*{8pt}
\caption{The multifractal spectrum of Example 2 with $|\Theta|=1,...,25$}
\end{figure}

\begin{table}[pt]
\tbl{The value of $y$ with the increase of $|\Theta|$ of Example 2.}
{\begin{tabular}{@{}cccccccc@{}} \toprule
& $y(\theta_{1})$&
$y(\theta_{1},\theta_{2})$ &
$y(\theta_{1},...,\theta_{3})$ &
$y(\theta_{1},...,\theta_{4})$ &
$y(\theta_{1},...,\theta_{5})$ &
$y(\theta_{1},...,\theta_{6})$ &
$...$ \\
$|\Theta|=2$\hphantom{00} & \hphantom{0}1.4650 & \hphantom{0}0.4650 \\
$|\Theta|=3$\hphantom{00} & \hphantom{0}1.5131 & \hphantom{0}0.9486 & \hphantom{0}0.5131 \\
$|\Theta|=4$\hphantom{00} & \hphantom{0}1.5415 & \hphantom{0}1.1358 & \hphantom{0}0.8229 &\hphantom{0}0.5415 \\
$|\Theta|=5$\hphantom{00} & \hphantom{0}1.5585 & \hphantom{0}1.2386 & \hphantom{0}0.9918 &\hphantom{0}0.7699 &\hphantom{0}0.5585 \\
$|\Theta|=6$\hphantom{00} & \hphantom{0}1.5688 & \hphantom{0}1.3036 & \hphantom{0}1.0991 &\hphantom{0}0.9152 &\hphantom{0}0.7400 &\hphantom{0}0.5688 \\
\botrule
\end{tabular}}
\end{table}

\begin{table}[pt]
\tbl{The value of $\widetilde{f}(y)$ with the increase of $|\Theta|$ of Example 2.}
{\begin{tabular}{@{}cccccccc@{}} \toprule
& $\widetilde{f}(y(\theta_{1}))$&
$\widetilde{f}(y(\theta_{1},\theta_{2}))$ &
$\widetilde{f}(y(\theta_{1},...,\theta_{3}))$ &
$\widetilde{f}(y(\theta_{1},...,\theta_{4}))$ &
$\widetilde{f}(y(\theta_{1},...,\theta_{5}))$ &
$\widetilde{f}(y(\theta_{1},...,\theta_{6}))$ &
$...$ \\
$|\Theta|=2$\hphantom{00} & \hphantom{0}0.6309 & \hphantom{0}0 \\
$|\Theta|=3$\hphantom{00} & \hphantom{0}0.5646 & \hphantom{0}0.5646 & \hphantom{0}0 \\
$|\Theta|=4$\hphantom{00} & \hphantom{0}0.5119 & \hphantom{0}0.6616 & \hphantom{0}0.5119 &\hphantom{0}0 \\
$|\Theta|=5$\hphantom{00} & \hphantom{0}0.4687 & \hphantom{0}0.6705 & \hphantom{0}0.6705 &\hphantom{0}0.4687 &\hphantom{0}0 \\
$|\Theta|=6$\hphantom{00} & \hphantom{0}0.4325 & \hphantom{0}0.6536 & \hphantom{0}0.7231 &\hphantom{0}0.6536 &\hphantom{0}0.4325 &\hphantom{0}0 \\
\botrule
\end{tabular}}
\end{table}

As can be seen from Fig. 3, the multifractal spectrum in this example is plotted by solid colored lines. The black dotted line is the case that $\lim_{|\Theta|\rightarrow \infty }$. In other words, the multifractal spectrum of of mass function with maximum Deng entropy when $n \rightarrow \infty$ can be approximated by a quadratic function $F(x)=-a(x-0.585)(x-1.585)$, where $a=\frac{4log_{2}C_{n}^{\frac{n}{2}}}{n}$. Proof is as follows. According to Eq.~(\ref{12}), Eq.~(\ref{14}) and Eq.~(\ref{15}).

\begin{equation}
   \lim_{n\rightarrow \infty }m(\theta_{1})=\frac{2^{1}-1}{C_{n}^{1}(2^{1}-1)+C_{n}^{2}(2^{2}-1)+...+C_{n}^{n}(2^{n}-1)}=\frac{1}{3^{n}-2^{n}}\,. \label{22}
\end{equation}

\begin{equation}
   \lim_{n\rightarrow \infty }m(\theta_{1},...,\theta_{\frac{n}{2}})=\frac{2^{\frac{n}{2}}-1}{C_{n}^{1}(2^{1}-1)+C_{n}^{2}(2^{2}-1)+...+C_{n}^{n}(2^{n}-1)}=\frac{2^{\frac{n}{2}}-1}{3^{n}-2^{n}}\,. \label{221}
\end{equation}

\begin{equation}
   \lim_{n\rightarrow \infty }m(\theta_{1},...,\theta_{n})=\frac{2^{n}-1}{C_{n}^{1}(2^{1}-1)+C_{n}^{2}(2^{2}-1)+...+C_{n}^{n}(2^{n}-1)}=\frac{2^{n}-1}{3^{n}-2^{n}}\,. \label{23}
\end{equation}

\begin{equation}
\begin{split}
   \lim_{n\rightarrow \infty }y(\theta_{1})=\frac{-log_{2}(3^{n}-2^{n})^{-1}}{-log_{2}(2^{n}-1)^{-1}}=\frac{log_{2}(2^{n})+log_{2}(\frac{3}{2}^{n}-1)}{log_{2}(2^{n}-1)}\approx 1.5850\,. \label{24}
\end{split}
\end{equation}

\begin{equation}
    \lim_{n\rightarrow \infty }y(\theta_{1},...,\theta_{\frac{n}{2}})=\frac{-log_{2}(\frac{2^{\frac{n}{2}}-1}{3^{n}-2^{n}})}{-log_{2}(2^{n}-1)^{-1}}=\frac{-log_{2}(2^{\frac{n}{2}})+log_{2}(\frac{3}{2}^{n}-1)}{log_{2}(2^{n}-1)}\approx1.0850\,. \label{251}
\end{equation}

\begin{equation}
    \lim_{n\rightarrow \infty }y(\theta_{1},...,\theta_{n})=\frac{-log_{2}(\frac{2^{n}-1}{3^{n}-2^{n}})}{-log_{2}(2^{n}-1)^{-1}}=0.5850\,. \label{25}
\end{equation}

\begin{equation}
    \lim_{n\rightarrow \infty }\widetilde{f}(y(\theta_{1})=\frac{log_{2}(C_{n}^{1})}{-log_{2}(2^{n}-1)^{-1}}\approx\frac{log_{2}(n)}{n}\approx0\,. \label{26}
\end{equation}

\begin{equation}
    \lim_{n\rightarrow \infty }\widetilde{f}(y(\theta_{1},...,\theta_{\frac{n}{2}}))=\frac{log_{2}(C_{n}^{\frac{n}{2}})}{-log_{2}(2^{n}-1)^{-1}}=\approx \frac{log_{2}(C_{n}^{\frac{n}{2}})}{n}\,. \label{261}
\end{equation}

\begin{equation}
    \lim_{n\rightarrow \infty }\widetilde{f}(y(\theta_{1},...,\theta_{n}))=\frac{log_{2}(C_{n}^{n})}{-log_{2}(2^{n}-1)^{-1}}\approx\frac{log_{2}(1)}{n}=0\,. \label{27}
\end{equation}

$$\left \{(y(\theta_{1}),\widetilde{f}(y(\theta_{1}))),(1.5850,0)\right \}$$ 
$$\left \{(y(\theta_{1},...,\theta_{n}),\widetilde{f}(y(\theta_{1},\theta_{n}))),(0.5850,0)\right \}$$

There are three points $(0.5850,0)$, $(1.0850,\frac{log_{2}(C_{n}^{\frac{n}{2}})}{n})$, $(1.5850,0)$, thus the quadratic function is calculated as $F(x)=-a(x-0.585)(x-1.585)$, where $a=\frac{4log_{2}C_{n}^{\frac{n}{2}}}{n}$.

\noindent \textbf {Example 3}: Consider a framework of discernment is $\Theta$, $|\Theta|=2,...,n$. A mass function averagely assigned on $2^{\Theta}$ is $m_{i}(A)=\frac{1}{2^{|\Theta|}-1}$, $i=2,...,n$. The multifractal spectrums is only one point $(1,1)$ in Fig. 4.

\begin{equation}
    y(\theta_{1})=...=y(\theta_{1},\theta_{2})=...=y(\theta_{1},...,\theta_{n})=\frac{-log_{2}(2^{n}-1)^{-1}}{-log_{2}(2^{n}-1)^{-1}}=1\,. \label{28}
\end{equation}

\begin{equation}
    \widetilde{f}(y(\theta_{1}))=...=\widetilde{f}(y(\theta_{1},...,\theta_{n}))=\frac{log_{2}(2^{n}-1)}{-log_{2}(2^{n}-1)^{-1}}=1\,. \label{29}
\end{equation}

\begin{figure}[th]
\centerline{\includegraphics[width=10cm]{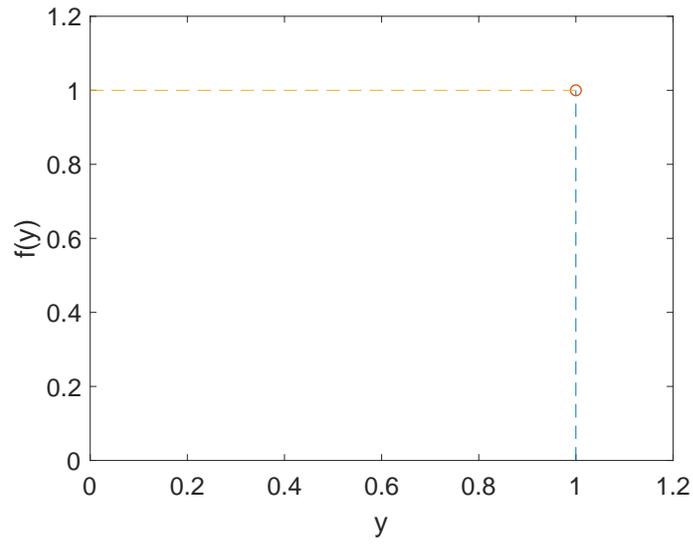}}
\vspace*{8pt}
\caption{The multifractal spectrum of Example 3 with $|\Theta|=1,...,25$}
\end{figure}

\noindent \textbf {Example 4}: Consider a framework of discernment is $\Theta$, $|\Theta|=2,...,n$. A mass function is $m_{i}(\Theta)=1$, $i=2,...,n$. The multifractal spectrums is only one point $(0,0)$ in Fig. 5.

\begin{equation}
    y(\Theta)=\frac{-log_{2}(1)^{-1}}{-log_{2}(2^{n}-1)^{-1}}=0\,. \label{30}
\end{equation}

\begin{equation}
    \widetilde{f}(y(\Theta))=\frac{log_{2}(1)}{-log_{2}(2^{n}-1)^{-1}}=0\,. \label{31}
\end{equation}

\begin{figure}[th]
\centerline{\includegraphics[width=10cm]{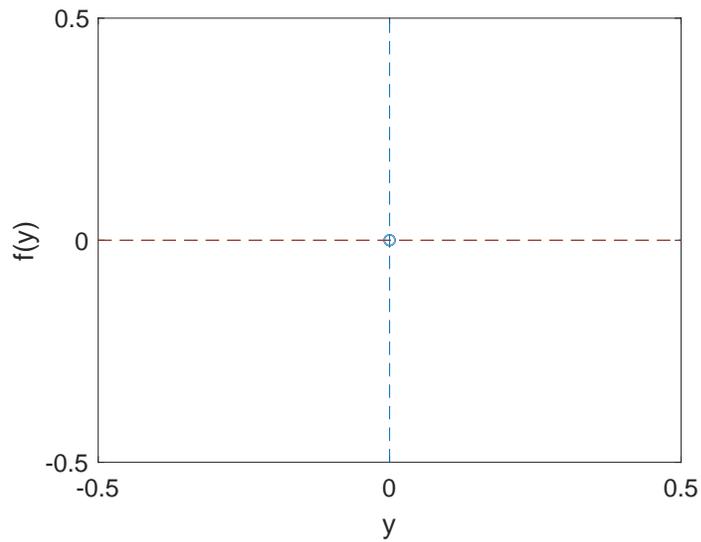}}
\vspace*{8pt}
\caption{The multifractal spectrum of Example 4 with $|\Theta|=1,...,25$}
\end{figure}

\section{Multifractal dimension of mass function}
\begin{definition}
The multifractal dimension of mass function is defined as follows.
\begin{equation}
     D_{\alpha}(m(A_{i}))=\frac{\frac{1}{1-\alpha }log_{2}\left [ \sum_{i}(\frac{m(A_{i})}{2^{|A_{i}|}-1|})^{\alpha }(2^{|A_{i}|}-1)\right ]}{log_{2}\sum_{i}(2^{|A_{i}|}-1)^{m(A_{i})^{\alpha}}}\,. \label{32}
\end{equation}
\end{definition}

\begin{theorem}
$ \lim_{\alpha\rightarrow 1}D_{\alpha}(m(A_{i}))=\frac{-\sum_{i} m(A_{i})log_{2}(\frac{m(A_{i})}{2^{|A_{i}|}-1})}{log_{2}\sum_{i}(2^{|A_{i}|}-1)^{m(A_{i})}}$, where the numerator is Deng entropy.
\end{theorem}

\begin{proof}
\begin{equation}
\begin{split}
    \lim_{\alpha \rightarrow 1}D_{\alpha}(m(A_{i}))=& \lim_{\alpha \rightarrow 1}\frac{\frac{\partial }{\partial \alpha} \left [ log_{2} \sum_{i}(\frac{m(A_{i})}{2^{|A_{i}|}-1|})^{\alpha }(2^{|A_{i}|}-1)\right]}{\frac{\partial }{\partial \alpha} \left [(1-\alpha) (log_{2}\sum_{i}(2^{|A_{i}|}-1)^{m(A_{i})^{\alpha}}) \right ]}\\=&\frac{ \sum_{i}\left [ (\frac{m(A_{i})}{2^{|A_{i}|}-1|})^{\alpha }(2^{|A_{i}|}-1)log_{2}(\frac{m(A_{i})}{2^{|A_{i}|}-1})\right ]}{(-log_{2}\sum_{i}(2^{|A_{i}|}-1)^{m(A_{i})^{\alpha}})+(1-\alpha)T}\\=&\frac{-\sum_{i} m(A_{i})log_{2}(\frac{m(A_{i})}{2^{|A_{i}|}-1})}{log_{2}\sum_{i}(2^{|A_{i}|}-1)^{m(A_{i})}} \,. \label{33}
\end{split}
\end{equation}
\end{proof}

Where $T$ is $\left [ log_{2}\sum_{i}(2^{|A_{i}|}-1)^{m(A_{i})^{\alpha}}\right ]^{'}$, which multipled by $(1-\alpha)$ tend to be 0 when $\alpha \rightarrow 1$.

\begin{theorem}
When mass function degenerates to probability distribution, the proposed multifractal dimension degenerates to Renyi information dimension.

\end{theorem}

\begin{proof}
\begin{equation}
\begin{split}
     D_{\alpha}(m_{i}\rightarrow P_{i})=&\frac{\frac{1}{1-\alpha }log_{2}\left [ \sum_{i}(\frac{p_{i}}{2^{1}-1})^{\alpha }(2^{1}-1)\right ]}{log_{2}\sum_{i}(2^{1}-1)^{p_{i}^{\alpha}}}\\=&\frac{\frac{1}{1-\alpha}log_{2}\sum_{i}p_{i}^{\alpha}}{log_{2}n}\\=&d_{\alpha}\,. \label{34}
\end{split}
\end{equation}
\end{proof}

\begin{theorem}
Let $\Theta$ has $n$ elements, a mass function with maximum Deng entropy is $m(A_{i})=\frac{2^{|A_{i}|-1}}{\sum_{i}2^{|A_{i}|}-1}$, $\lim_{n\rightarrow \infty }D_{\alpha}=1.585$, where $ \alpha \in R$.

\end{theorem}

\begin{proof}
\begin{equation}
\begin{split}
     \lim_{n\rightarrow \infty } D_{ \alpha}(m(A_{i}))=&\frac{\frac{1}{1-\alpha }log_{2}\left [ \sum_{i}(\frac{\frac{2^{|A_{i}|}-1}{\sum_{i}(2^{|A_{i}|}-1)}}{2^{|A_{i}|}-1|})^{\alpha }(2^{|A_{i}|}-1)\right ]}{log_{2}\sum_{i}(2^{|A_{i}|}-1)^{\frac{2^{|A_{i}|}-1}{\sum_{i}(2^{|A_{i}|}-1)}^{\alpha}}}\\=&\frac{\frac{1}{1-\alpha}log_{2}(3^{n}-2^{n})^{-\alpha}(\sum_{i}2^{|A_{i}|}-1)}{log_{2}(2^{|\Theta|}-1)}\\=&\frac{\frac{1}{1-\alpha}log_{2}(3^{n}-2^{n})^{1-\alpha}}{log_{2}2^{n}}\\\approx&1.585\,. \label{35}
\end{split}
\end{equation}
\end{proof}

Where $\lim_{n\rightarrow \infty }(2^{|A_{i}|}-1)^{\frac{2^{|A_{i}|}-1}{\sum_{i}(2^{|A_{i}|}-1)}^{\alpha}}=(2^{|A_{i}|}-1)^{(\frac{2^{|A_{i}|}-1}{3^{n}-2^{n}})^{\alpha}}\approx((2^{|A_{i}|}-1)^{0})^{\alpha}=1$ and there are  $(2^{|\Theta|}-1)$ terms.

\begin{theorem}
Let $\Theta$ has $n$ elements, a mass function is $m(\theta_{i})=\frac{1}{|\Theta|}=\frac{1}{n}$, $i=1,2,...,n$. $\lim_{n\rightarrow \infty }D_{\alpha}=1$, where $ \alpha \in R$.

\end{theorem}

\begin{proof}
\begin{equation}
\begin{split}
     \lim_{n\rightarrow \infty } D_{ \alpha}(m(A_{i}))=&\frac{\frac{1}{1-\alpha }log_{2} \sum_{i}(\frac{n^{-1}}{2^{|\theta_{i}|}-1|})^{\alpha }(2^{|\theta_{i}|}-1)}{log_{2}\sum_{i}(2^{|\theta_{i}|}-1)^{(n^{-1})^{\alpha}}}\\=&\frac{\frac{1}{1-\alpha }log_{2} \sum_{i}({n^{-1}})^{\alpha }}{log_{2}\sum_{i}^{n}(1)}\\=&\frac{\frac{1}{1-\alpha}log_{2}n^{1-\alpha}}{log_{2}n}\\=&1\,. \label{36}
\end{split}
\end{equation}
\end{proof}

It should be noted that in this theorem, mass function degenerates to probability and the proposed multifractal dimension degenerates to Renyi information dimension. This theorem can be written as: the Renyi information dimension of uniform distribution is a fixed point 1 no matter what $\alpha$ is.
~\\

\noindent \textbf {Example 5}: Given a framework of discernment is $\Theta=\left \{\theta_{1},\theta_{2},\theta_{3} \right \}$. A mass function is $m(\theta_{1})=0.2,\ m(\theta_{2},\theta_{3})=0.8$. The results of multifractal dimension are shown in Table 3. The detailed calculations are as follows.

When $\alpha = 1$, according to Theorem 1, 

\begin{equation}
    D_{\alpha}=\frac{-(0.2\times log_{2}(\frac{0.2}{2^{1}-1})+0.8\times log_{2}(\frac{0.8}{2^{2}-1}))}{log_{2}((2^{1}-1)^{0.2}+(2^{2}-1)^{0.8})}=1.1249\,. \label{37}
\end{equation}

When $\alpha = 2$,

\begin{equation}
    D_{\alpha}=\frac{-\frac{1}{1-2}\times log_{2}\left [(\frac{0.2}{2^{1}-1})^{2}\times(2^{1}-1)+(\frac{0.8}{2^{2}-1})^{2}\times(2^{2}-1) \right ]}{log_{2}\left [ ((2^{1}-1)^{0.2})^{2}+((2^{2}-1)^{0.8})^{2}\right ]}=0.7163\,. \label{38}
\end{equation}

\begin{table}[pt]
\tbl{The result of Example 5.}
{\begin{tabular}{@{}cccccccc@{}} \toprule
& $\alpha=3$&
$\alpha=9$ &
$\alpha=15$ &
$\alpha=21$ &
$\alpha=27$ &
$\alpha=33$ &
$...$ \\

$D_{\alpha}$\hphantom{00} & \hphantom{0}0.5759 & \hphantom{0}0.2390 & \hphantom{0}0.1472 &\hphantom{0}0.1060 &\hphantom{0}0.0828 &\hphantom{0}0.0679 \\
\botrule
\end{tabular}}
\end{table}

As can be seen from Table 3, $D_{\alpha}$ goes to zero as $\alpha$ increasing.

\noindent \textbf {Example 6}: Given a framework of discernment is $\Theta$, $|\Theta|=2,...,n$. A mass function is $m_{i}(\Theta)=1$, $i=2,...,n$. The results are show in Table 4 and Fig. 6. Detailed calculations are as follows.

\begin{equation}
\begin{split}
     D_{\alpha}=&\frac{\frac{1}{1-\alpha}log_{2}(\frac{m(\Theta)}{2^{|\Theta|}-1})^{\alpha}(2^{|\Theta|}-1)}{log_{2}\left [ (2^{|\Theta|}-1)^{m(\Theta)}\right ]^{\alpha}}\\=& \frac{\frac{1}{1-\alpha}log_{2}(\frac{1}{2^{n}-1})^{\alpha-1}}{log_{2}\left [ (2^{n}-1)^{1}\right ]^{\alpha}}\\=&\frac{log_{2}(2^{n}-1)}{\alpha log_{2}(2^{n}-1)}\\=&\frac{1}{\alpha}\,. \label{39}
\end{split}
\end{equation}

From Eq.~(\ref{39}), the multifractal dimensions of $m(\Theta)=1$ are only related to the order $\alpha$ and as $\alpha$ increases, $D_{\alpha}$ decreases.

\begin{table}[pt]
\tbl{The result of Example 6.}
{\begin{tabular}{@{}ccccccccc@{}} \toprule
$|\Theta|=n $ $\setminus $ $\alpha$ & $1$&
$4$ &
$7$ &
$10$ &
$13$ &
$16$ &
$19$ & \\
$D_{\alpha}$\hphantom{00} & \hphantom{0}1 & \hphantom{0}0.25 & \hphantom{0}0.1429 &\hphantom{0}0.1 &\hphantom{0}0.0769 & \hphantom{0}0.0625 &\hphantom{0}0.0526 \\
\botrule
\end{tabular}}
\end{table}

\begin{figure}[th]
\centerline{\includegraphics[width=10cm]{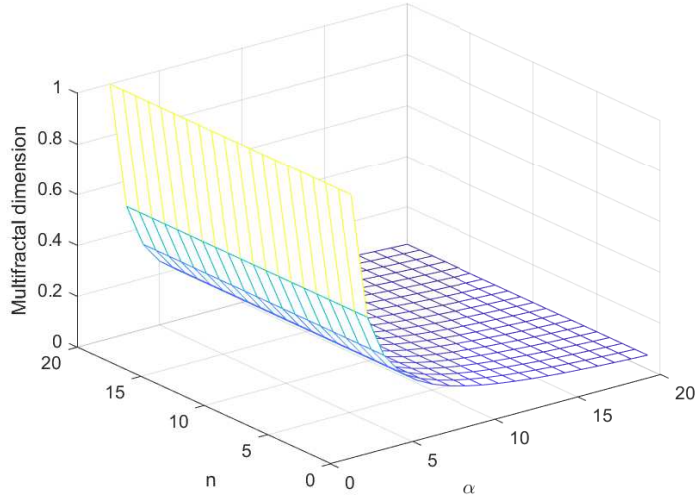}}
\vspace*{8pt}
\caption{The change of multifractal dimension of Example 6}
\end{figure}

\noindent \textbf {Example 7}: Given a framework of discernment is $\Theta$, $|\Theta|=2,...,n$. A  mass function is $m(\theta_{i})=\frac{1}{|\Theta|}$. The result is calculated as follows.

\begin{equation}
    d_{\alpha}=\frac{\frac{1}{1-\alpha}log_{2}\sum_{i}^{n}(\frac{1}{n})^{\alpha}}{log_{2}n}=\frac{log_{2}n}{log_{2}n}=1\,. \label{40}
\end{equation}

In this example, the mass function degenerates to probability distribution. According to Theorem 2, the proposed dimension degenerates to Renyi information dimension. From Eq.~(\ref{40}), the multifractal dimension of $m(\theta_{i})=\frac{1}{n}$ is 1. It doesn't depend on  $\alpha$ and $n$. This example illustrates Theorem 4.

\noindent \textbf {Example 8}: Given a framework of discernment is $\Theta$, $|\Theta|=2,...,n$. A mass function is $m(A_{i})=\frac{1}{2^{|\Theta|}-1)}$. The results are shown in Table 5 and Fig. 7.

\begin{table}[pt]
\tbl{The result of Example 8.}
{\begin{tabular}{@{}cccccccccc@{}} \toprule
$|\Theta| $ $\setminus $ $\alpha$ & $1$&
$5$ &
$9$ &
$13$ &
$17$ &
$21$ &
$25$ & 
$29$ & \\
$2$\hphantom{00} & \hphantom{0}1.1850 & \hphantom{0}0.5682 & \hphantom{0}0.3413 &\hphantom{0}0.2370 &\hphantom{0}0.1804 & \hphantom{0}0.1455 &\hphantom{0}0.1218 & \hphantom{0}0.1048 \\
$4$\hphantom{00} & \hphantom{0}1.3811 & \hphantom{0}0.9707 & \hphantom{0}0.8180 &\hphantom{0}0.7146 &\hphantom{0}0.6321 & \hphantom{0}0.5637 &\hphantom{0}0.5061 & \hphantom{0}0.4462 \\
$6$\hphantom{00} & \hphantom{0}1.4520 & \hphantom{0}1.0988 & \hphantom{0}1.0023 &\hphantom{0}0.9518 &\hphantom{0}0.9138 & \hphantom{0}0.8814 &\hphantom{0}0.8521 & \hphantom{0}0.8251 \\
$8$\hphantom{00} & \hphantom{0}1.4798 & \hphantom{0}1.1433 & \hphantom{0}1.0599 &\hphantom{0}1.0265 &\hphantom{0}1.0062 & \hphantom{0}0.9911 &\hphantom{0}0.9787 & \hphantom{0}0.9679 \\
$10$\hphantom{00} & \hphantom{0}1.4911 & \hphantom{0}1.1620 & \hphantom{0}1.0788 &\hphantom{0}1.0491 &\hphantom{0}1.0333 & \hphantom{0}1.0230 &\hphantom{0}1.0156 & \hphantom{0}1.0097 \\
$12$\hphantom{00} & \hphantom{0}1.4959 & \hphantom{0}1.1724 & \hphantom{0}1.0865 &\hphantom{0}1.0568 &\hphantom{0}1.0417 & \hphantom{0}1.0324 &\hphantom{0}1.0261 & \hphantom{0}1.0215 \\
$14$\hphantom{00} & \hphantom{0}1.4981 & \hphantom{0}1.1795 & \hphantom{0}1.0907 &\hphantom{0}1.0603 &\hphantom{0}1.0450 & \hphantom{0}1.0357 &\hphantom{0}1.0296 & \hphantom{0}1.0251 \\
$16$\hphantom{00} & \hphantom{0}1.4991 & \hphantom{0}1.1851 & \hphantom{0}1.0973 &\hphantom{0}1.0624 &\hphantom{0}1.0467 & \hphantom{0}1.0373 &\hphantom{0}1.0311 & \hphantom{0}1.0266 \\
$18$\hphantom{00} & \hphantom{0}1.4995 & \hphantom{0}1.1897 & \hphantom{0}1.0960 &\hphantom{0}1.0640 &\hphantom{0}1.0480 & \hphantom{0}1.0384 &\hphantom{0}1.0320 & \hphantom{0}1.0264 \\
$20$\hphantom{00} & \hphantom{0}1.4998 & \hphantom{0}1.1935 & \hphantom{0}1.0980 &\hphantom{0}1.0653 &\hphantom{0}1.0490 & \hphantom{0}1.0392 &\hphantom{0}1.0327 & \hphantom{0}1.0280 \\

\botrule
\end{tabular}}
\end{table}

\begin{figure}[th]
\centerline{\includegraphics[width=10cm]{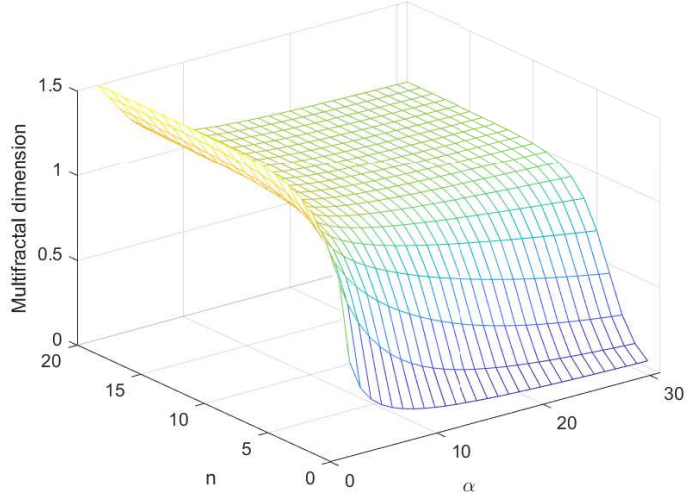}}
\vspace*{8pt}
\caption{The change of multifractal dimension of Example 8}
\end{figure}

As can be seen from Table 5, the multifractal dimension is getting closer and closer to 1 when the size of framework of discernment bigger with higher order $\alpha$. This rule of change can be seen more intuitively in Fig. 7. Actually, we can come to this conclusion directly from the definition as follows.

\begin{equation}
    \begin{split}
        \lim_{n \rightarrow \infty ,\alpha\rightarrow \infty  }D_{G\alpha}=&\frac{\frac{1}{1-\alpha }log_{2}\left [ \sum_{i}(\frac{\frac{1}{2^{n}-1}}{2^{|A_{i}|}-1|})^{\alpha }(2^{|A_{i}|}-1)\right ]}{log_{2}\sum_{i}(2^{|A_{i}|}-1)^{(\frac{1}{2^{n}-1})^{\alpha}}}\\ \approx&\frac{\frac{1}{1-\alpha}log_{2}\sum_{i}^{2^{n}-1}(\frac{1}{2^{n}-1})^{\alpha}}{log_{2}\sum_{i}^{2^{n}-1}(1)^{\alpha}}\\ \approx&\frac{\frac{1}{1-\alpha}log_{2}(2^{n}-1)^{1-\alpha}}{log_{2}(2^{n}-1)}\\=&1\,. \label{41}
    \end{split}
\end{equation}

\noindent \textbf {Example 9}: Given a framework of discernment is $\Theta$, $|\Theta|=2,...,n$. A mass function with maximum Deng entropy is $m(A_{i})=\frac{2^{|A_{i}|}-1}{\sum_{i}(2^{|A_{i}|}-1)}$. The multifractal dimensions are shown in Table 6 and Fig. 8.

\begin{table}[pt]
\tbl{The result of Example 9.}
{\begin{tabular}{@{}ccccccccc@{}} \toprule
$|\Theta| $ $\setminus $ $\alpha$ & $1$&
$4$ &
$7$ &
$10$ &
$13$ &
$16$ &
$19$ & \\
$2$\hphantom{00} & \hphantom{0}1.1752 & \hphantom{0}0.5809 & \hphantom{0}0.3473 &\hphantom{0}0.2441 &\hphantom{0}0.1878 & \hphantom{0}0.1526 &\hphantom{0}0.1285 \\
$4$\hphantom{00} & \hphantom{0}1.4699 & \hphantom{0}1.1962 & \hphantom{0}0.8893 &\hphantom{0}0.6589 &\hphantom{0}0.5124 & \hphantom{0}0.4172 &\hphantom{0}0.3515 \\
$6$\hphantom{00} & \hphantom{0}1.5516 & \hphantom{0}1.4904 & \hphantom{0}1.4065 &\hphantom{0}1.2899 &\hphantom{0}1.1437 & \hphantom{0}0.9919 &\hphantom{0}0.8575 \\
$8$\hphantom{00} & \hphantom{0}1.5747 & \hphantom{0}1.5611 & \hphantom{0}1.5457 &\hphantom{0}1.5275 &\hphantom{0}1.5052 & \hphantom{0}1.4770 &\hphantom{0}1.4406 \\
$10$\hphantom{00} & \hphantom{0}1.5816 & \hphantom{0}1.5784 & \hphantom{0}1.5750 &\hphantom{0}1.5715 &\hphantom{0}1.5677 & \hphantom{0}1.5637 &\hphantom{0}1.5593 \\
$12$\hphantom{00} & \hphantom{0}1.5838 & \hphantom{0}1.5830 & \hphantom{0}1.5822 &\hphantom{0}1.5814 &\hphantom{0}1.5806 & \hphantom{0}1.5797 &\hphantom{0}1.5789 \\
$14$\hphantom{00} & \hphantom{0}1.5846 & \hphantom{0}1.5844 & \hphantom{0}1.5842 &\hphantom{0}1.5840 &\hphantom{0}1.5838 & \hphantom{0}1.5836 &\hphantom{0}1.5834 \\
$16$\hphantom{00} & \hphantom{0}1.5848 & \hphantom{0}1.5848 & \hphantom{0}1.5847 &\hphantom{0}1.5847 &\hphantom{0}1.5846 & \hphantom{0}1.5846 &\hphantom{0}1.5845 \\
$18$\hphantom{00} & \hphantom{0}1.5849 & \hphantom{0}1.5849 & \hphantom{0}1.5849 &\hphantom{0}1.5849 &\hphantom{0}1.5849 & \hphantom{0}1.5848 &\hphantom{0}1.5848 \\
$20$\hphantom{00} & \hphantom{0}1.5849 & \hphantom{0}1.5849 & \hphantom{0}1.5849 &\hphantom{0}1.5849 &\hphantom{0}1.5849 & \hphantom{0}1.5849 &\hphantom{0}1.5849 \\
\botrule
\end{tabular}}
\end{table}

\begin{figure}[th]
\centerline{\includegraphics[width=10cm]{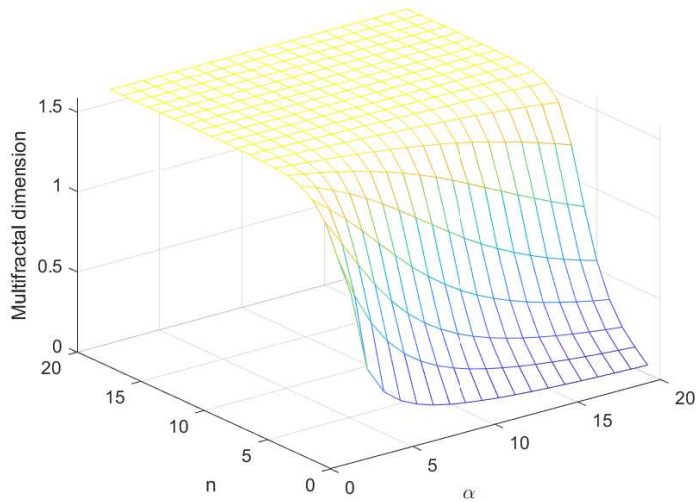}}
\vspace*{8pt}
\caption{The change of multifractal dimension of Example 9}
\end{figure}

It can be seen from Table 6 and Fig. 8 that when the order $\alpha$ is determined, the multifractal dimension increases and finally tends to a constant with the increase of the framework of discernment. When the size of $|\Theta|$ is determined, the multifractal dimension decreases with the improvement of $\alpha$. More specifically, when the size of the framework of discernment is small, the multifractal dimension will decrease to 0 finally. When $n$ get bigger, the tendency to decrease gets smaller and smaller until it doesn't change and remains $1.5850$. This example demonstrates Theorem 3.

\section{Conclusion}

In this paper, the multifractal spectrum of mass function is defined. Three special assignments are studied. The multifractal spectrum of mass function with maximum Deng entropy approximates quadratic function $F=-a(x-0.585)(x-1.585)$, where $a=\frac{4log_{2}C_{n}^{\frac{n}{2}}}{n}$. The multifractal spectrum of mass function with total uncertainty is the point (0,0) and with average assignment on power set is the point (1,1). Another important work is that we propose multifractal dimension of mass function. It noted that when mass function degenerates to probability and is distributed averagely, the proposed dimension degenerates to Renyi information dimension and is a constant 1. In addition, the multifractal dimension of mass function with maximum Deng entropy goes to 1.585 with the condition $\alpha \rightarrow \infty,\ |\Theta| \rightarrow \infty$. Other interesting properties are discussed. The changes of proposed dimension with differnet parameters $\alpha$ and the size of the framework of discernment $\Theta$ are shown by numerical examples. 

\section*{Acknowledgements}

The work is partially supported by National Natural Science Foundation of China (Grant No. 61973332), JSPS Invitational Fellowships for Research in Japan (Short-term).

\bibliographystyle{ws-fnl}
\bibliography{ws-fnl}

\end{document}